\documentclass[conference]{IEEEtran}
\usepackage{cite}

\ifCLASSINFOpdf
\else
\fi

\usepackage{booktabs}   
\usepackage{subcaption} 
\usepackage{dblfloatfix}
\usepackage{tikz}
\usepackage{bp-listings}
\usepackage{tikz_liveness}
\usepackage{paper_plots}
\usepackage{pgfplots}
\usepackage{verbatim}
\usepackage{amsmath}
\usepackage{amsthm}
\usepackage{amsfonts}
\usepackage{amssymb}
\usepackage{multicol}
\usepackage[hidelinks]{hyperref}
\usepackage{enumitem}
\usepackage{xspace}
\usepackage{multirow}
\usepackage{threeparttable}
\usepackage{orcidlink}
\usepackage{bbm}
\theoremstyle{definition}
\newtheorem{theorem}{Theorem}
\newtheorem{claim}{Claim}
\newtheorem{definition}{Definition}
\newtheorem{proposition}{Proposition}
\newtheorem{assumption}{Assumption}
\newtheorem{example}{Example}

\newcommand{\aref}[1]{\hyperref[#1]{Appendix~\ref*{#1}}}

\definecolor{codegreen}{rgb}{0,0.6,0}
\definecolor{codegray}{rgb}{0.5,0.5,0.5}
\definecolor{codepurple}{rgb}{0.58,0,0.82}
\definecolor{backcolour}{rgb}{0.98,0.98,0.96}

\lstdefinelanguage{BPPy}
{morekeywords={interrupt, block, request, waitFor, Block, Request, mustFinish, localReward, yield, while, def, for, in, range, from, import, return, if, else, lambda, sum, filter, class, super},
sensitive=false,
morecomment=[l]{\#},
morecomment=[s]{/*}{*/},
morestring=[b]",
morestring=[b]',
}

\lstdefinestyle{mystyle}{
    language=BPPy,
    backgroundcolor=\color{backcolour},   
    commentstyle=\color{codegreen},
    keywordstyle=\color{blue},
    keywordstyle=[5]\ttfamily\scriptsize
    numberstyle=\tiny\color{codegray},
    stringstyle=\color{codepurple},
    basicstyle=\ttfamily
    \lst@ifdisplaystyle\footnotesize\fi,
    breakatwhitespace=false,         
    breaklines=true,                 
    captionpos=b,                    
    keepspaces=true,                 
    numbers=none,                    
    numbersep=5pt,                  
    showspaces=false,                
    showstringspaces=false,
    showtabs=false,                  
    tabsize=2
}
\newcommand{\bt}[1]{\lstinline[language=,breaklines=false,prebreak=]{req\_#1}}
\newcommand{\li}[2][]{{\footnotesize\lstinline[breaklines=false,prebreak=,#1]{#2}}}

\author{\IEEEauthorblockN{Tom Yaacov\orcidlink{0000-0002-0565-6506}}
\IEEEauthorblockA{Computer Science Department\\
Ben-Gurion University of the Negev\\
tomya@post.bgu.ac.il}
\and
\IEEEauthorblockN{Achiya Elyasaf\orcidlink{0000-0002-4009-5353}}
\IEEEauthorblockA{Software and Information\\ Systems Engineering Department\\
Ben-Gurion University of the Negev\\
achiya@bgu.ac.il}
\and
\IEEEauthorblockN{Gera Weiss\orcidlink{0000-0002-5832-8768}}
\IEEEauthorblockA{Computer Science Department\\
Ben-Gurion University of the Negev\\
geraw@cs.bgu.ac.il}}

\begin{document}
%
\title{Making Liveness Requirements Executable in Behavioral Programming}
\title{Adding Liveness Requirements to the Behavioral Programming Executable Specification Language}

\title{Adding and Executing  Liveness Requirements in Behavioral Programs}

\title{Bringing Behavioral Programs to Life: \\ Specifying and Executing Liveness Requirements}

\title{Keeping Behavioral Programs Alive: \\ Specifying and Executing Liveness Requirements}

\maketitle



%


\maketitle

\begin{abstract}
One of the benefits of using executable specifications such as Behavioral Programming (BP) is the ability to align the system implementation with its requirements. This is facilitated in BP by a protocol that allows independent implementation modules that specify what the system may, must, and must not do. By that, each module can enforce a single system requirement, including negative specifications such as ``don't do X after Y.'' The existing BP protocol, however, allows only the enforcement of safety requirements and does not support the execution of liveness properties such as ``do X \emph{at least} three times.'' To model liveness requirements in BP directly and independently, we propose idioms for tagging states with ``must-finish,'' indicating that tasks are yet to be completed. We show that this idiom allows a direct specification of known requirements patterns from the literature. We also offer semantics and two execution mechanisms, one based on a translation to B\"uchi automata and the other based on a Markov decision process (MDP). The latter approach offers the possibility of utilizing deep reinforcement learning (DRL) algorithms, which bear the potential to handle large software systems effectively. This paper presents a qualitative and quantitative assessment of the proposed approach using a proof-of-concept tool. A formal analysis of the MDP-based execution mechanism is given in an appendix.
\end{abstract}


\IEEEpeerreviewmaketitle
\section{introduction}
\label{sec:introduction}

The vision of executable specifications~\cite{delahaye_common_2013, gajski_specc_2012, johnsen_abs_2010} is to streamline software development by providing methods for transforming individual requirements directly into implementation modules. This helps to reduce the discrepancy between what was intended and what the system actually does and allows developers to focus on capturing and implementing requirements clearly and concisely, eliminating the need to write additional code that is not directly related to them~\cite{kastens2005,schubert2006}. In contrast, traditional programming practices often involve indirect translation of requirements into code, which can increase complexity, make it difficult to ensure that the code accurately reflects the intended behavior, and hinder maintenance. 

In practice, executable specification methods face the challenge of balancing expressiveness and the ability to execute. For instance, using an idiom like ``win the game'' in chess-playing software may seem appealing, but we lack an execution mechanism that directly implements such requirements. Conversely, we can limit ourselves to requirements that specify the exact response to each opponent's move. However, this approach necessitates a complex translation of requirements and often results in over-specification. Consequently, designers of executable specification frameworks aim to create specification languages that allow requirements to be modeled directly while still allowing the model to be executed without requiring manual code translation.

This paper focuses on the Behavioral Programming (BP) executable specification paradigm. 
This paradigm has been successfully applied in various applications, such as reactive IoT building~\cite{elyasaf_context-based_2018}, a fully functional nano-satellite~\cite{bar-sinai_scenario_2019}, an autonomous rover~\cite{katz_--fly_2019}, and tool suite for model-driven testing~\cite{bar-sinai_provengo_2023}, showcasing its versatility and potential.
Alignment of requirements to implementation was identified as a key design goal in the original proposal BP~\cite{harel_programming_2010, elyasaf_context-oriented_2021}. Nevertheless, we show in this paper that this objective has not been fully realized due to the limitations of current idioms in expressing liveness properties. Consequently, BP cannot directly express and execute requirements such as ``the robot must eventually reach its target'' or ``the robot must be charged infinitely often.'' To address this limitation, this paper proposes the ``must-finish'' idiom used to flag states that indicate that some goals have not yet been achieved in these states. Furthermore, we demonstrate how this idiom allows for a compact specification of requirements using a dataset of common requirement patterns~\cite{dwyer_patterns_1999}.

As said above, proposing modeling idioms is only one part of the equation in executable specifications. An efficient execution mechanism that adheres to these idioms is also necessary. 
We propose two methods to run programs with liveness specifications: by transforming a BP program into a Generalized Büchi Automaton (GBA) or by transforming it into a Markov Decision Process (MDP). These two models represent the program's behavior through their respective semantics and enable algorithms for execution. The GBA-based method defines liveness in the automata's acceptance criteria. The MDP-based method defines liveness by designing a reward function that embodies the desired system behavior. GBA allows synthesis algorithms to produce executions, while MDP allows for methods based on reinforcement learning.

The two methods proposed in this paper are general as they are not tied to specific algorithms for solving a GBA or MDP. Our evaluation even shows that, in the case of MDPs, an exact solution is not mandatory. We demonstrate the potential of the MDP-based approach in using approximate solutions, such as deep neural networks, to learn from systems with high-dimensional state space. This opens up the possibility of using executable specifications in domains where state space exploration is not feasible.

The paper is structured as follows: In \autoref{sec:motivation}, we introduce BP and provide the rationale for the concepts discussed in this paper. \autoref{sec:adding-liveness} extends BP's idioms and abstract semantics. In \autoref{sec:develop-by-verification}, we discuss the challenges and issues with current BP practices. \autoref{sec:ltl-patterns} demonstrates how the ``must-finish'' idiom aids in aligning models with specifications using known requirement patterns. In \autoref{sec:enforcing-liveness-game}, we present the GBA-based approach for enforcing liveness properties. Subsequently, \autoref{sec:enforcing-liveness} describes the MDP-based execution mechanism, and a formal proof of its correctness is provided in \aref{app:a}. \autoref{sec:drl} showcases the potential of the MDP-based approach in learning systems with large state space using neural networks. In \autoref{sec:examples}, we discuss the effectiveness of our methods in handling multiple liveness requirements. We discuss related work in \autoref{sec:related-work}, and conclude with a short discussion in \autoref{sec:conclusion}.

\section{an example-driven introduction to bp and the need for adding liveness}
\label{sec:motivation}
Behavioral Programming (BP)~\cite{harel_behavioral_2012} is a model-based programming paradigm that allows users to specify reactive systems' behavior directly aligned with how the developers perceive the system requirements. 
In BP, users specify scenarios, called \emph{b-threads}, representing behaviors the system should include or avoid. Each scenario is standalone and is usually self-contained, concerning itself with only one specific aspect of the system behavior, typically a single requirement. 
An application-agnostic execution engine interprets and interweaves these scenarios at runtime to produce cohesive system behavior consistent with the system requirements.
Specifically, the execution mechanism is based on a synchronization protocol proposed by~\cite{harel_programming_2010}. The protocol consists of each b-thread submitting a statement before selecting each event that the b-program produces. When a b-thread reaches a point where it is ready to submit a statement, it synchronizes with its peers and declares which events it requests, which events it waits for (but does not request), and which events it blocks (forbids from happening). 
After submitting the statement, the b-thread is paused. When all b-threads have submitted their statements, we say the b-program has reached a \emph{synchronization point}. Then, the event arbiter picks a single event that was requested and was not blocked and resumes all b-threads that requested or waited for that event. The other b-threads remain paused, and their statements are used in the next synchronization point. The process is repeated throughout the execution of the program.


We begin with an illustrative example adapted from the well-known level-crossing benchmark~\cite{leveson_safety_1987} to make these concepts more concrete. Although this model is over 35 years old, it is still widely used for specification and analysis purposes~\cite{elyasaf_what_2023, liu_-penda_2014, ghazel_customizable_2016, mazzeo_systematic_1997}, with over 40 citations in 2023. We will now use this benchmark to introduce BP, and later, we will add liveness requirements and explain the need for liveness execution semantics in BP. 

The benchmark describes a controller for a gate at a railway crossing---an intersection of multiple railway lines and a road at the same level. Each railway line has a sensor that signals the controller whenever the train approaches, enters, or leaves the crossing zone. Based on the signals, the barriers are raised and lowered, ensuring the safety of the trains, i.e., that a train cannot be in the crossing zone while the barriers are up. More formally, the system requirements are as follows:

\begin{enumerate}[label=\textbf{R\arabic*.}, ref=R\arabic*]
    \item \label{itm:req-1} For each railway, the sensors system dictates the exact event order: train approaching, entering, and leaving. Additionally, there is no overlapping between successive train passages.
    \item \label{itm:req-2} The barriers are lowered when a train approaches and then raised as soon as possible.
    \item \label{itm:req-3} A train may not enter while barriers are up.
    \item \label{itm:req-4} The barriers may not be raised while a train is in the intersection zone. The intersection zone is the area between the approaching and leaving sensors.
\end{enumerate}
There is no train at the intersection zone at system initialization, and the barriers are raised. 

\begin{lstlisting}[
float=htp,
caption={A BP program that specifies the requirements for the level-crossing benchmark~\cite{elyasaf_what_2023}. Each b-thread is aligned with a single requirement.},
label={lst:LC-base}
]
def req_1(railway):
  while True:
    yield sync(waitFor=Approaching(railway))
    yield sync(request=Entering(railway), 
               block=Approaching(railway))
    yield sync(request=Leaving(railway), 
               block=Approaching(railway))

def req_2():
  while True:
    yield sync(waitFor=any_approaching)
    yield sync(request=Lower)
    yield sync(request=Raise)

def req_3():
  while True:
    yield sync(waitFor=Lower,block=any_entering)
    yield sync(waitFor=Raise)

def req_4(railway):
  while True:
    yield sync(waitFor=Approaching(railway))
    yield sync(waitFor=Leaving(railway),block=Raise)
\end{lstlisting}
\vspace{-12pt}

\autoref{lst:LC-base} shows a b-program (a set of b-threads) that implements these requirements. The code is written using BPpy~\cite{yaacov_bppy_2023}, a framework for BP in Python. In BPpy, each b-thread is implemented as a Python generator---a function that can pause itself and pass data back to its caller at any point, using the \li{yield} idiom. It can then be resumed when re-invoked with the \li{send} method. The statements submitted by each b-thread are structured as \li{sync} class instances containing events or event sets labeled by the arguments \li{request}, \li{block}, \li{waitFor}. The execution mechanism starts by independently invoking each b-thread and awaiting its statement yield. Once all statements have been collected, an event is selected, and the program resumes its execution based on the aforementioned synchronization protocol.


There are four numbered b-threads, each aligned with its corresponding requirement. The \textit{first b-thread}, \bt{1}, waits for a train to approach the crossing zone on a specific railway and then dictates the exact sequence of events for each railway, as defined in the first requirement. The dictation is achieved by requesting the entering and leaving events in the correct order. The blocking idiom ensures that there will be no overlapping between successive train passages. The \textit{second b-thread} specifies the barriers behavior, waiting for a train to approach, then requesting to lower the barriers, and finally raising them as soon as possible. The \textit{third b-thread} blocks trains from entering while the barriers are up. Finally, the \textit{last b-thread} blocks the raising of the barriers while there is a train in the intersection zone.

Since no b-thread requests the approaching event, no event is selected at the first synchronization point, and the program halts. This is acceptable since the requirements do not specify how many trains should approach and on which railway. Moreover, these requirements do not specify the number of railways in the system. To address these ambiguities, we concrete this system to a specific level-crossing intersection with three railways, one for each train type: passengers, freight, and maintenance. Additionally, the train schedule dictates the following traffic requirements:

\begin{enumerate}[label=\textbf{R\arabic*.}, ref=R\arabic*]
    \setcounter{enumi}{4}
    \item \label{itm:req-5} Passenger trains may approach at any time. 
    \item \label{itm:req-6} Freight trains approach at least three times.
    \item \label{itm:req-7} Maintenance trains approach at least three times.
\end{enumerate}

The code in \autoref{lst:lc-trains} reflects these requirements. The first requirement states that passenger trains approach at any time, as denoted by the \li{while True} loop. The subsequent two requirements specify that the freight and maintenance trains approach \textit{at least} three times each. While the code specifies \textit{exactly} three times, these b-threads are still aligned with the requirements as they request three approaching events and do not block additional ones. This implies that other b-threads (e.g., those added later) may request additional approaching events.

\begin{lstlisting}[
% float=htpb,
caption={Additional b-threads that specify the trains requirements (\ref{itm:req-5}-\ref{itm:req-7}).},
label={lst:lc-trains}
] 
def req_5():
  while True:
    yield sync(request=Approaching("Passenger"))

def req_6(): 
  for i in range(3):
    yield sync(request=Approaching("Freight"))

def req_7():
  for i in range(3):
    yield sync(request=Approaching("Maintenance"))
\end{lstlisting}
\vspace{-5pt}

Upon execution, all b-threads are executed until they reach the first synchronization point, which is indicated by their first \li{yield sync} command. At this point, the last three b-threads request for their approaching events, while the remaining b-threads only include \li{waitFor} and \li{block} in their statements. The event arbiter then selects an event that has been requested and is not blocked. In this case, it selects one of the following events: \li{Approaching("Passenger")}, \li{Approaching("Freight")}, or \li{Approaching("Maintenance")}. Based on the selected event, all b-threads that requested or waited for it are resumed and advanced to their next synchronization point. For example, if the selected event is \li{Approaching("Freight")}, then \bt{6} is resumed and requests its approaching event for the second time, and \bt{2} advances as well and requests to lower the barrier. \bt{1} and \bt{4} have three instances, one for each railway; therefore, their ``Freight'' copies are also advanced. The rest of the b-threads remain in place, and their statements continue to the next round. The process is repeated and generates a sequence of system events.

Unlike most programming paradigms, BP does not force developers to choose a single behavior for the system to follow. For instance, this b-program does not dictate the order in which trains approach the intersection zone. Instead, the system can choose any behavior that complies with the b-threads and, therefore, meets the system requirements. This is in contrast to a standard programming language that dictates a specific sequence of actions. Since this forces programmers to specify beyond what is stated in the requirements, traditional programming paradigms are prone to over-specification, while BP avoids it.

We now present the main observation that motivated our research: BP lacks liveness execution semantics, which impedes the development and analysis of software systems. This observation will be illustrated by further discussing the level-crossing example. Suppose that after running the initial version of the system for a certain period, the scheduler introduces the following additional requirement:
 
\begin{enumerate}[label=\textbf{R\arabic*.}, ref=R\arabic*]
\setcounter{enumi}{7}
    \item \label{itm:req-8} A maintenance train must approach between two freight trains approaching.
\end{enumerate}

While we can modify b-threads \bt{6} and \bt{7} by introducing new conditions and statements, the BP paradigm motivates us to add a new b-thread for each new requirement. Thus, we add a new b-thread, presented in \autoref{lst:LC-r7}, that waits for a freight train to approach. Then, it blocks the next freight train from approaching until a maintenance train approaches. Note that this b-thread can be added (and removed) without affecting the other b-threads. This is an example of a purely additive change, where the system's behavior is altered to comply with a new requirement without changing the existing specification.

\begin{lstlisting}[
% float=htpb,
caption={A b-thread specifying \ref{itm:req-8}.},
label={lst:LC-r7}
]
def req_8():
  while True:
    yield sync(waitFor=Approaching("Freight"))
    yield sync(waitFor=Approaching("Maintenance"), 
               block=Approaching("Freight"))
\end{lstlisting}
\vspace{-3pt}

Examining the updated b-program's possible runs, i.e., sequences of events, reveals that all desired runs are possible. However, note that some undesired runs are also possible. Consider the order in which trains approach the crossing zone, presented as a sequence of letters, where $F$, $M$, and $P$ denote the approaching of freight, maintenance, and passenger trains, respectively. While $FMMFMF(P)^\omega$ and $FMFMFM(P)^\omega$ are desired runs, $MMMF(P)^\omega$ and $FMMMF(P)^\omega$ are not, since the freight train did not approach \textit{at least} three times. We use $(\hdots)^\omega$ to mark the periodic part of the run. Although this issue became evident with the addition of b-thread \bt{8}, it is rooted in the (mis)alignment of b-thread \bt{6} to \ref{itm:req-6}, as the b-thread does not enforce a freight train to approach \textit{at least} three times. This issue cannot be resolved without modifying BP semantics since specifying that something will occur at least three times is a liveness property---a property of the type ``something good will eventually occur'' as defined by Lamport~\cite{lamport_proving_1977}. However, the current semantics of BP only support the specification of safety properties.

This simple example effectively demonstrates the significance of incorporating liveness properties into BP models. It highlights the necessity of considering liveness even when modeling basic requirements, such as enforcing a minimum occurrence of an event.

\section{formal semantics for bp with liveness}
\label{sec:adding-liveness}

In this section, we present the \emph{must-finish} idiom and add it to BP semantics. We illustrate how enabling b-threads to use this idiom, as demonstrated in the b-thread in \autoref{lst:LC-live} with the \li{mustFinish=True} flag, allows for the specification of liveness requirements.


\vspace{-6pt}
\begin{lstlisting}[
label={lst:LC-live},
caption={A b-thread specifying that ``Three freight trains \textit{must eventually} approach'' using the \li{mustFinish} idiom.},
float=htbp
]
def req_6_mod(): 
  for i in range(3):
    yield sync(request=Approaching("Freight"), 
               mustFinish=True)
\end{lstlisting}
\vspace{-8pt}

We begin by formally defining the syntax and semantics of BP and making the necessary modifications to accommodate liveness specifications. The semantics of BP, as described by Harel et al.~\cite{harel_programming_2010}, define b-threads and their collective execution using labeled transition systems (LTSs). To designate states as \textit{must-finish}, we suggest introducing a labeling function to these LTSs. 

Recall that an LTS is defined as a tuple $\langle S,E,\rightarrow,init\rangle$, where $S$ is a set of states, $E$ is a set of events, $\rightarrow \subseteq S \times E \times S$ is a transition relation, and $init \in S$ is the initial state~\cite{keller_formal_1976}. The runs of such a transition system are sequences of the form $s^0 \xrightarrow{e^1} s^1 \xrightarrow{e^2} \dots \xrightarrow{e^i} s^i\dots$ where $s^0=init$, and $\forall i\in\mathbb{N}$, $s^i\in S$, $e^i\in E$, and $s^{i-1} \xrightarrow{e^i} s^i$.

We define b-threads, as defined by Harel et al.~\cite{harel_programming_2010}, only adding the labeling function $L$ that we propose:

\begin{definition}{(b-thread).}
A b-thread is a tuple $\langle S,E,\rightarrow,init,B,R,L\rangle$ where $\langle S,E,\rightarrow,init\rangle$ forms an LTS, $R\colon S\rightarrow 2^E$ associates a state with the set of events requested by the b-thread in that state, $B\colon S\rightarrow 2^E$ associates a state with the set of events blocked in that state, and $L \colon S\rightarrow \{0,1\}$ is a labeling function that indicates if the state is \emph{must-finish}.
\end{definition}

\begin{example}
\label{example:thread-with-liveness}
The b-thread described in \autoref{lst:LC-live} consists of four states: one for each iteration $i$ in the loop ($s_1,s_2,s_3$); and a terminal state, $s_4$, reached when the loop ends. At each iteration, the b-thread requests the \li{Approaching("Freight")} event, does not block any events and marks the state as \li{mustFinish}. Formally, for all $i \in \{1,2,3\}$:
$R(s_{i}) = \{\text{\li{Approaching("Freight")}}\}$, $B(s_{i}) = \emptyset$, and $L(s_{i}) = 1$.
Once the loop ends, the b-thread reaches the terminal state, $s_4$, where it does not submit any statements to the event arbiter. Hence, vacuously, we have:
$B(s_4) = \emptyset$, $R(s_4) = \emptyset$, and $L(s_4) = 0$.

\end{example}

We also modify the definition of a b-program to include b-threads that define the $L$ function:
\begin{definition}{(b-program).}
\label{def:b-program}
A b-program is a set of b-threads $\{\langle S_i,E_i,\rightarrow_i,init_i,R_i,B_i,L_i\rangle\}_{i=1}^n$
\end{definition}

Next, we proceed to define the semantics of the model in terms of their runs:

\begin{definition}{(run of a b-program).}
\label{def:run}
A run of a b-program, $\{\langle S_i,E_i,\rightarrow_i,init_i,R_i,B_i,L_i\rangle\}_{i=1}^n$, is a run of the LTS $\langle S,E,\rightarrow,init\rangle$, where $S = S_1\times\dots\times S_n$, $E=\bigcup_{i=1}^nE_i$, $init = \langle init_1,\dots,init_n\rangle$, and $\rightarrow$ includes a transition $\langle s_1,\dots ,s_n\rangle \xrightarrow{e} \langle s_1',\dots ,s_n'\rangle$ if and only if
\vspace{-2pt}
$$e \in \bigcup_{i=1}^n R_i(s_i) \bigwedge e \notin \bigcup_{i=1}^n B_i(s_i)$$ \vspace{-10pt} and 
$$\bigwedge\limits_{i=1}^n ((e \in E_i \Rightarrow s_i \xrightarrow{e} s_i' )\wedge (e \notin E_i \Rightarrow s_i = s_i' )).$$
\end{definition}

Here also, the definition of a run of a b-program has not changed, except for the inclusion of the b-threads $L$ function. 

The main addition to the semantics is the live run definition:

\begin{definition}{(live b-program run).}
\label{def:liverun}
We say that a run of a b-program is live if, for each b-thread $i$ and time $t$, there exists $t'>t$ such that 
$L_i(s_i^{t'}) = 0$.
\end{definition}

In other words, the definition states that a run of a b-program is considered live if, at any given time and for every b-thread, there exists a future time where the local must-finish state of that thread becomes zero. This means that no b-thread remains stuck in a must-finish state indefinitely.

\begin{example}
Considering the b-thread defined in \autoref{example:thread-with-liveness} and its implementation in \autoref{lst:LC-live}, we can observe that all runs which include a minimum of three instances of \li{Approching("Freight")} are live runs. Conversely, other runs do not meet the criteria for liveness due to the b-thread becoming trapped in some state $s_i$ where it requests an event that is never received while its \li{mustFinish} flag is on.
\end{example}

Our notion of liveness is based on the concept of hot cycles introduced by Harel et al.~\cite{harel_model-checking_2011} and further refined by Bar-Sinai~\cite{bar-sinai_extending_2020}. Hot cycles are characterized by labeling states as either ``hot'' or ``cold'' and using this labeling to identify two types of hot cycles: ``b-thread hot'' cycles, where a b-thread remains in a hot state throughout the entire cycle, and ``b-program hot'' cycles, where at least one b-thread is in a hot state in every state of the cycle. Conversely, from the complementary perspective of ``cold cycles'', a cycle is considered ``b-thread cold'' if every b-thread becomes cold during the cycle. By replacing ``cold'' with ``not must-finish'', this definition leads to ``live runs.'' ``B-program hot cycles'' can be replicated using ``must-finish" by adding a b-thread that constantly waits for all other b-threads to become cold and keeps its must-finish flag active until that condition is met.

Harel et al.~\cite{harel_model-checking_2011} and Bar-Sinai~\cite{bar-sinai_extending_2020} did not provide an execution mechanism for the model. Instead, they suggested verifying that the program does not have these cycles. This facilitates a ``develop-by-verification" approach, in which the program is repeatedly checked for violations and modified accordingly. However, this approach has inherent limitations, which we will discuss in the next section.

\section{why not develop-by-verification?}
\label{sec:develop-by-verification}
The ``develop-by-verification'' approach, proposed, e.g., in~\cite{harel_model-checking_2011,harel_non-intrusive_2014}, involves developers defining their liveness properties and verifying the system's correctness. This approach hinges on the fact that when a property is violated, the verifier provides a counter-example, i.e., a sequence of events that lead to the violation. Based on the provided example, the user can add one or more b-threads that monitor these behaviors and prevent the violation. 

For instance, applying verification on the level-crossing b-program introduced in \autoref{sec:motivation} to detect the presented liveness issue yielded the counter-example:
\li{Approaching("Freight")},\li{Lower},$($\li{Approaching("Passenger")} 
,\li{Entering("Passenger")},\li{Leaving("Passenger")}$)^\omega$.
Upon examination of this sequence, it became apparent that there is a potential for the passenger railway to create a liveness violation, resulting in the starvation of freight trains and preventing them from passing through. A possible solution for starvation involves scheduling to provide access to the denied resource for all components. Thus, we first tried mitigating this issue by adding the b-thread in \autoref{lst:LC-fix1}.

\begin{lstlisting}[
% float=htpb,
caption={A solution to the starvation of freight trains. It requires that a non-passenger train must approach between two passenger trains approaching in the first six rounds. },
label={lst:LC-fix1}
]
def avoid_freight_starvation():
  for i in range(6):
    yield sync(waitFor=Approaching("Passenger"))
    yield sync(waitFor=any_approaching, 
               block=Approaching("Passenger"))
\end{lstlisting}

However, as highlighted in \autoref{sec:motivation}, starvation represents just one aspect of the program violations, as \bt{8} can still prevent freight trains from crossing three times. Further, adding the b-thread in \autoref{lst:LC-fix1} inadvertently introduced a new violation, leading to a deadlock scenario upon completing non-passenger train crossings. Consequently, we made some additional attempts and added the b-thread detailed in \autoref{lst:LC-fix2}. It introduces adjusted scheduling that ensures each type of train approaches exactly once in the initial three rounds by preventing trains that have already approached from doing so again. This means the it allows unlimited passenger trains to cross once non-passenger trains have finished crossing.

The b-thread provided in \autoref{lst:LC-fix2} resolved all violations, and the program now meets all requirements. It, however, also restricts behaviors that do not violate any of the requirements. This over-specification can be an issue as we want to allow all feasible system behaviors. Allowing such flexibility becomes crucial, for example, when considering development cycles where future requirements might prohibit (block) the behaviors the system is left with or when optimizing the implemented system further based on some additional objective. Further, the implementation that includes the scheduling b-thread is not directly aligned with its requirements, which makes it harder to understand and maintain. 

\begin{lstlisting}[
% float=htpb,
caption={A solution to the liveness violation caused by the misalignment of b-thread \bt{6} and \ref{itm:req-6}. It requires each type of train to approach exactly once in the first three rounds. },
label={lst:LC-fix2}
]
def fix_scheduling_issues():
  for i in range(3):
    blocked = []
    e = yield sync(waitFor=any_app)
    blocked.append(e)
    e = yield sync(waitFor=any_app, block=blocked)
    blocked.append(e)
    yield sync(waitFor=any_app, block=blocked)
\end{lstlisting}

One solution to the problems we demonstrated with the ``develop-by-verification'' approach is to analyze the program's state space and identify a pattern that generalizes the violation detected by the model checker. Some verification techniques involve explicit traversal of the program's state space in search of violations, allowing us to examine it once the verification process concludes. Searching the entire space can be hard in certain scenarios and infeasible in others. Additionally, deciphering the various violating patterns can be a complex task even if we could readily obtain a program graph and look only at the vicinity of the detected violation~\cite{chapman_learning_2015,groce_error_2006,beer_explaining_2012,copty_efficient_2003}.



In conclusion, the key takeaway from this section is that while model-checkers can effectively identify liveness violations and provide counter-examples that can sometimes be used to fix systems, this approach can be challenging to implement and often leads to over-specification and models that do not align with the requirements. This paper proposes a correct-by-construction alternative where liveness requirements are directly modeled within the b-threads. By ensuring that the execution mechanism respects these specifications, we guarantee that the system's runs include all the required behaviors and nothing else. There is no need for a post-process of analyzing and fixing errors. All the requirements can be represented directly in the model, fully aligned with how users perceive them.


\section{using the dwyer patterns to measure the expressiveness and usability}
\label{sec:ltl-patterns}
To assess the effectiveness of the must-finish idiom in expressing liveness specifications in executable models, this section focuses on mapping common specifications to b-threads using the new idiom. We use a benchmark of specification patterns proposed by Dwyer et al.~\cite{dwyer_patterns_1999}. The dataset includes patterns that are common in industrial specifications. The collection encompasses 55 patterns, with 18 classified as liveness specifications (i.e., contain a liveness property). The supplementary material for this paper provides a full mapping of all the patterns to BP using the proposed \li{mustFinish} idiom proposed in this paper. For this section, we limited the presentation to patterns that include liveness properties, as the BP idioms before our addition were already sufficient for expressing safety specifications. The Dwyer patterns are organized by kind, e.g., absence or existence of properties, and are ordered by scope, e.g., globally or before an additional property. An example pattern of kind 'existence' with the scope 'after' is ``$p$ occurs after $q$,'' where $p$ and $q$ are pattern parameters, representing some non-temporal propositions. The b-thread in  \autoref{lst:dwyer-8} models this property. 
To ensure a fair pattern representation, in this section, we assume that each system event is a boolean assignment over a set of atomic propositions. This assumption aligns with standard formal specification practices. Additionally, BP offers extensions beyond the standard protocol that enable the modeling of systems with events containing multiple variables~\cite{katz_--fly_2019, harel_executing_2019}. These extensions provide a rich set of composable constraints, allowing for the description of both desired and undesired variable assignments. Thus, the b-thread in \autoref{lst:dwyer-8} waits for an event in which \li{q} holds and \li{p} does not, using the waited-for statement \li{And(q, Not(p))}. It then uses the new idiom to state that, eventually, $p$ must occur.

\begin{lstlisting}[
label={lst:dwyer-8},
caption={A b-thread specifying that ``$p$ occurs after $q$.''},
% float=htbp
]
def pattern():
    yield sync(waitFor=And(q, Not(p)))
    yield sync(waitFor=p, mustFinish=True)
\end{lstlisting}

To see how the must-finish idiom complements existing BP idioms for expressing undesired behaviors, consider the pattern ``$s$ responds to $p$, after $q$ until $r$.''
In simple terms, it specifies that after the occurrence of $q$, whenever $p$ happens, $s$ must eventually occur, and this condition persists until $r$ takes place. We provide a b-thread that models this property in \autoref{lst:dwyer-30}. This b-thread waits for \li{q} to hold without an immediate occurrence of \li{r}, then proceeds to wait for \li{p} without an \li{s}, which immediately responds to it. Lastly, it uses the must-finish idiom and the blocking idiom to specify that \li{s} must eventually occur until \li{r} does.

\begin{lstlisting}[
label={lst:dwyer-30},
caption={A b-thread specifying that ``$s$ responds to $p$, after $q$ until $r$.'' },
% float=htbp
]
def pattern():
  while True:
    e = yield sync(waitFor=And(q, Not(r)))
    if (not e.p) or e.s:
      e = yield sync(waitFor=Or(r, And(p, Not(s))))
    while not e.r:
      yield sync(waitFor=s,block=r,mustFinish=True)
      e = yield sync(waitFor=Or(r, And(p, Not(s))))
\end{lstlisting}

In \autoref{sec:enforcing-liveness-game}, we illustrate how we can view a b-thread as a B\"uchi Automaton (BA). Due to the determinism of b-program transitions, as defined in \autoref{def:run}, the resulting automaton is deterministic. However, this may pose a challenge as some properties, such as ``eventually always $p$,'' cannot be defined using deterministic BAs~\cite{manna_hierarchy_1990}.
To allow a direct translation of such requirements, we introduce non-determinism within b-threads. This is easily achieved by adding a new property to system events, not constrained by any b-thread. Consequently, non-deterministic choices can be made inside b-threads using the last chosen event property. For instance, the b-thread in \autoref{lst:dwyer-0} implements the above pattern. It first enters a loop that must be finished sometime in the future non-deterministically. Then, \li{p} is blocked forever.

\begin{lstlisting}[
label={lst:dwyer-0},
caption={A b-thread specifying that ``eventually always $p$.''
% GERA: Explain the b-threads in words. 
},
% float=htbp, 
]
def pattern():
  while True:
    e = yield sync(waitFor=true(), mustFinish=True)
    if non_deterministic_choice(e):
        break
  yield sync(block: Not(p))
\end{lstlisting}


All patterns presented here and the remaining Dwyer patterns available in the supplementary material were verified using the model checker available in BPpy. To ensure the correctness of each implemented b-thread and its ability to prohibit all undesired runs, the verifier examined a b-program containing the respective b-thread alongside an auxiliary b-thread that continuously requests all events. The verification process assessed the program with the assumption that it always eventually exits the must-finish states. This means that the program's specification was evaluated under the condition that if this assumption holds, the pattern must also be satisfied.

The conclusion of our experiment with the Dwyer patterns is that the new \li{mustFinish} idiom effectively expresses all the patterns in the dataset with small, simple b-threads, similar to the few examples in this section.

\section{two liveness enforcing execution mechanisms}

Since BP focuses on providing an executable specification language, we must deliver effective ways to execute specifications that contain the \li{mustFinish} idiom. Here, we present two mechanisms and explain the theory that supports their validity. We will examine their effectiveness in later sections.

\subsection{GBA-based execution mechanism}
\label{sec:enforcing-liveness-game}

In this section, we explain how automata theory can be used to maintain the liveness of b-programs. We transform the b-program into a Generalized B\"uchi Automaton (GBA), essentially a single-player game. The game's solution is then used to guide the event selection mechanism in selecting events that preserve the liveness of the system.

A Generalized B\"uchi Automaton (GBA) is defined as a tuple $\langle S, E, \to, I, F\rangle$, where $S$ is a finite set of states, $E$ is a finite set of events, $\to \subseteq S \times E \times S$ is the transition relation, $I \subseteq S$ is the set of initial states, and $F = \{F_1, \ldots, F_k\}$, where each $F_i$ is a set of accepting states. A run is considered accepting if it visits at least one state from each set in $F$ infinitely often.

%
A GBA induces a game that involves a single player placing a token on an initial state and moving it along the edges of the automaton. The game's objective is to find a path that visits each set of accepting states infinitely often. To achieve this objective, the player must follow a winning strategy that helps them make informed decisions about transitioning from one state to another. The ultimate goal is to reach the acceptance sets repeatedly and indefinitely. If the player fails to accomplish this goal, they will lose the game.

The GBA game can be solved using different techniques. One such method is the SCC-based algorithm~\cite{renault_three_2013}. This algorithm solves single-player games over generalized-B\"uchi acceptance conditions by finding the game automaton's strongly connected components (SCCs). It splits the SCCs into two groups: winning SCCs for the player and losing SCCs. The player's winning strategy is then to always move to a winning SCC.
The algorithm uses graph traversal techniques, such as depth-first search, to identify the SCCs and their win/lose status.
Once the SCCs have been identified, the algorithm calculates their acceptance status and the status of their outgoing transitions to separate them into winning and losing sets. Finally, it provides the player's winning strategy, always moving to a winning SCC. This algorithm is straightforward and efficient for solving single-player games with generalized B\"uchi acceptance conditions. For the experimental section of this paper, we used an optimized version of this algorithm implemented using the Spot library~\cite{duret-lutz_spot_2016}. 

The translation of a b-program to a GBA is:

\begin{definition}{(b-program liveness GBA).}
\label{def:buchi}
Let $\langle S, E, \to, init \rangle$ be an LTS of a b-program as in \autoref{def:run}.
A liveness GBA for this b-program is $\langle S,E,\rightarrow,init,\{F_i\}_{i=1}^n \rangle$ where $F_i = \{s \colon L_i(s) = 0\}$ for each b-thread $i$.
\end{definition}

By transforming the b-program in this manner, we can analyze the desired behaviors of the system and ensure it is satisfied during its operation. The solution to the game represents the optimal way for the system to satisfy the requirements, including any liveness requirements, while avoiding any undesirable behaviors. This provides a solid and robust framework for ensuring live operation.

We use the winning strategy to produce an event selection strategy for the system, guaranteeing that all the required sets will be visited and the winning condition will be satisfied.
Our implementation for the event arbiter for a b-program avoids events that take us out of the winning-states set and non-deterministically selects from the rest. Doing so ensures that the system operates according to the desired behaviors and never enters any must-finish regions it cannot leave. It also guarantees that all live runs remain possible, avoiding the over-specification phenomena mentioned in~\ref{sec:develop-by-verification}.

\subsection{MDP-based execution mechanism}
\label{sec:enforcing-liveness}
We now propose an alternative approach for obtaining an event selection mechanism using Markov decision processes (MDPs). In this approach, we formulate the problem as an MDP, where the states represent the different configurations of the system, and the actions correspond to the selectable events. 
We craft a reward function that reflects the desired behaviors to ensure the system meets the desired liveness requirements. By solving the Bellman equations, we obtain a policy that enforces liveness. We then introduce a liveness-preserving event selection strategy that utilizes this policy.

For simplicity, in this section, we focus on scenarios where there is only one b-thread in the system with a liveness requirement, referred to as the must-finish flag, and the labeling function for this b-thread is denoted as $L$. This does not limit the generality of our work, as multiple liveness requirements can always be expressed as a single one~\cite{baier_principles_2008}. 
Further, in \autoref{sec:examples}, we discuss how this approach can be generalized and allow multiple liveness requirements.


The following definition specifies how a b-program is converted into an MDP. An MDP is represented as a tuple with four elements:  $S$ represents the system states, $E$ the events it can execute, $R$ maps state-event-state transitions to real rewards, and $P$ maps state-event pairs to the probability of transitioning from one state to another given an event. The use of MDP allows reinforcement learning methods to find the desired policy for the system, as elaborated below.

The definition focuses on the reward function of the MDP representation of a b-program. This reward function, defined as $R\colon S\times E\times S \to \mathbb{R}$, reflects the desirability of each transition. This reward is designed to ensure liveness, meaning it captures the desired behavior of the b-program towards achieving its goals and satisfying its requirements.

\begin{definition}{(b-program liveness MDP).}
\label{def:bpmdp}
Let $\langle S, E, \to, init \rangle$ be the LTS of a b-program as in \autoref{def:run}. The b-program liveness MDP is $\langle S, E, R, P\rangle$ where
$$P(s'|s,e) = \begin{cases}
1 & \text{ if } s \xrightarrow{e} s'; \\
0 & \text{ otherwise; }
\end{cases}
$$
\vspace{-10pt}
and
$$R(s,e,s') = 
\left\{
	\begin{array}{ll}
		-1  & \text{if } L(s) = 0 \wedge L(s') = 1;\\
		1  & \text{if } L(s) = 1 \wedge L(s') = 0 ;\\
		0 & \text{otherwise.}
	\end{array}
\right.$$
\end{definition}
The reason for defining the rewards this way is to penalize runs, called ``live locks'', in which a b-program is caught in an infinite cycle whose states are labeled as \textit{must-finish}. We next use the standard action-value function approach to obtain a strategy that avoids such runs. 

The action-value function, $Q^\pi(s, e)$ of a reward function $R$, estimates the expected cumulative future reward obtained from $(s, e)$ when following a policy $\pi\colon S \to E$. 
The optimal action-value function $Q^*(s, e) = max_{\pi}Q^\pi(s, e)$ that provides maximal values in all states satisfies the well-known Bellman equation:
$Q^*(s, e) = \sum_{s'}P(s'|s, e)\Big(R(s,e,s') + \gamma max_{e'}Q_i^*(s', e')\Big)$, where $\gamma \in (0,1)$ is a discount factor that ensures that the accumulated reward of an infinite run is bounded. 
$Q^*(s, e)$ can be found using algorithms such as value iteration~\cite{sutton_introduction_1998}, Q-learning~\cite{watkins_q-learning_1992}, and deep Q-learning~\cite{mnih_human-level_2015}. 

To differentiate between live and non-live b-program runs, we introduce the notion of $Q^*$-compatible runs (stated in \autoref{def:Q-compatible run}) and a sampling distribution that models an event-selection strategy that chooses uniformly from this set (\autoref{def:Q-compatible policy}). These definitions enforce liveness requirements in the b-program execution, as we later show.

\begin{definition}{($Q^*$-comp. run).} 
\label{def:Q-compatible run}
Given $Q^*\colon S \times E \to \mathbb{R}$, the optimal action-value function that estimates the maximal expected future reward $R$, we define a run $s^0 \xrightarrow{e^0} s^1 \xrightarrow{e^1} \cdots $ of a b-program to be a $Q^*$-compatible run, if for each time $t$
\begin{equation}
\label{eqn:Q-compatible-condition}
\sum_{t'=0}^{t-1}R( s^{t'} ,e^{t'}, s^{t'+1} ) + Q^*(s^t,e^t) > -1
\end{equation}
\end{definition}

\begin{definition}{(sampling distribution for $Q^*$-comp. runs).} 
\label{def:Q-compatible policy}
We consider the following sampling distribution over the $Q^*$-compatible runs: For each state $s^t$, uniformly choose an event $e^t$ from all events such that Equation~\eqref{eqn:Q-compatible-condition} is satisfied.
\end{definition}

This sampling distribution can serve as an event selection strategy. Such a mechanism can ensure liveness properties during the b-program execution, 
as we show by proving in \aref{app:a} that a live b-program run is $Q^*$-compatible and that a $Q^*$-compatible b-program run is almost surely live:

\begin{theorem}
\label{theorem:ab}
A live run is $Q^*$-compatible, and a $Q^*$-compatible run is almost surely live.
\end{theorem}

Practically, this guarantees that if there is a way to generate a live run, our $Q^*$-compatible event selection mechanism will choose one of the possible live runs. For specifications that are not solvable, i.e., no run satisfies all liveness and safety requirements together, our mechanism will find an event that it cannot proceed from, informing the specifiers that their requirements contain conflicts.

\section{does the approach scale with drl?}
\label{sec:drl}
The preceding section showed how the GBA and MDP-based approaches can be leveraged to enforce liveness in b-programs. However, a naive implementation of these methods requires a direct translation and an exhaustive exploration of the entire program's state space. Such explicit translation may become infeasible as systems grow in size and complexity. To overcome this, this section examines a variation of the MDP-based approach based on deep reinforcement learning (DRL). In this variation, the action-value function is learned directly through trial and error, i.e., by executing the b-program many times instead of translating it to an explicit MDP. Additionally, neural networks are employed to approximate the action-value function being sought, with the program's state space, represented by the local variables of the b-threads, serving as input to the network. 
This approach, shown to be effective in learning from high-dimensional state spaces~\cite{mnih_human-level_2015}, opens up the possibility of utilizing executable specifications in domains where state space exploration is not feasible.

For demonstration, we took the level crossing example presented in \autoref{sec:motivation} and removed all b-threads irrelevant to the presented liveness issue. Specifically, to evaluate the DRL-based technique, our focus centered on a refined b-program consisting of b-threads \bt{6}, \bt{7}, and \bt{8}. Further, this example was parameterized in three dimensions to increase its state space and complexity: (1) freight and maintenance trains approach (at least) \emph{n} times instead of three; (2) there are \emph(m) maintenance railways, not just a single one; (3) a maintenance train must approach anywhere between \emph{k} freight trains approaching instead of two. The code developed for this evaluation is available in the supplementary material. 

We focus on identifying a deterministic event selection mechanism capable of generating a single execution trace that meets the system's requirements. Specifically, we employed the Maskable PPO algorithm~\cite{huang_closer_2022}, also implemented in~\cite{hill_stable_2018}, with a standard multilayer perception (MLP) network with two hidden layers of size 64. The results of this experiment, conducted using an NVIDIA GeForce RTX 4090 GPU, are available in \autoref{tab:drl-results-1}. For each parameter configuration, we computed the state space size and the time it took for the algorithm to converge to a valid deterministic policy. Entries where the state space size is colored in \textcolor{gray!75}{gray} indicate that the measurement is estimated due to memory limitations (200GB RAM). These \textcolor{gray!75}{gray} entries also mark complexities for which a direct translation is no longer suitable. We observe that while the problem size is increased, the algorithm's learning remains relatively stable. This stability can be attributed to the fact that the increasing values of $n$ and $k$ do not impact the number of program variables, whereas increasing $m$ linearly increases the number of variables. 
As a result, the complexity of the agent's neural network does not experience substantial changes. 
This underscores the potential benefits of the DRL approach in scenarios where system complexity is primarily driven by the range of values that program variables can assume rather than their quantity.
\vspace{-5pt}

\begin{table}[!ht]
    \centering
    \captionsetup{font=scriptsize}
    \scriptsize
    \setlength{\tabcolsep}{3pt} 
\begin{threeparttable}
\begin{tabular}{c|c|c c c c c c}
\multicolumn{2}{c}{} & \multicolumn{6}{c}{\bfseries $\boldsymbol{n}$} \\
    \bfseries $\boldsymbol{m}$ &	\bfseries $\boldsymbol{k}$ &  \bfseries $\boldsymbol{50}$ & \bfseries \bfseries $\boldsymbol{100}$ & \bfseries \bfseries $\boldsymbol{150}$ & \bfseries \bfseries $\boldsymbol{200}$ & \bfseries $\boldsymbol{250}$ & \bfseries $\boldsymbol{300}$ \\
        \hline
        \multirow{4}{*}{$\boldsymbol{1}$} & $\boldsymbol{1}$ & 0.003/2.8 & 0.010/7.5 & 0.023/6.3 & 0.041/6.0 & 0.063/6.3 & 0.091/6.7 \\
                             & $\boldsymbol{2}$ & 0.006/2.7 & 0.023/5.8 & 0.051/7.0 & 0.091/6.6 & 0.141/3.7 & 0.203/26.5 \\
                             & $\boldsymbol{3}$ & 0.008/6.8 & 0.034/5.6 & 0.075/7.3 & 0.134/0.6 & 0.209/0.7 & 0.301/3.3 \\
                             & $\boldsymbol{4}$ & 0.011/2.3 & 0.044/6.1 & 0.098/0.5 & 0.175/4.6 & 0.273/6.7 & 0.394/6.5 \\
        \hline
        \multirow{4}{*}{$\boldsymbol{2}$} & $\boldsymbol{1}$ & 0.201/6.1 & 1.555/6.4 & 5.187/18.4 & 12.22/10.3 & 23.78/17.3 & \textcolor{gray!75}{34}/18.6 \\
                             & $\boldsymbol{2}$ & 0.476/6.0 & 3.737/1.2 & 12.53/0.8 & 29.61/7.3 & \textcolor{gray!75}{49}/5.7 & \textcolor{gray!75}{78}/8.7\\
                             & $\boldsymbol{3}$ & 0.625/6.7 & 4.930/0.4 & 16.56/7.7 & \textcolor{gray!75}{37}/4.5 & \textcolor{gray!75}{65}/7.2 & \textcolor{gray!75}{102}/4.7 \\
                             & $\boldsymbol{4}$ & 0.754/6.7 & 5.974/3.9 & 20.10/0.6 & \textcolor{gray!75}{45}/1.0 & \textcolor{gray!75}{79}/7.2 & \textcolor{gray!75}{124}/6.1 \\
        \hline
        \multirow{4}{*}{$\boldsymbol{3}$} & $\boldsymbol{1}$ & 13.66/6.4 & \textcolor{gray!75}{50}/9.5 & \textcolor{gray!75}{105}/6.9 & \textcolor{gray!75}{187}/7.4 & \textcolor{gray!75}{292}/25.1 & \textcolor{gray!75}{422}/8.4 \\
                             & $\boldsymbol{2}$ & 32.35/7.3 & \textcolor{gray!75}{119}/7.1 & \textcolor{gray!75}{251}/4.0 & \textcolor{gray!75}{442}/8.5 & \textcolor{gray!75}{691}/1.5 & \textcolor{gray!75}{998}/10.0 \\
                             & $\boldsymbol{3}$ & 39.25/3.3 & \textcolor{gray!75}{144}/6.7 & \textcolor{gray!75}{304}/0.9 & \textcolor{gray!75}{536}/1.2 & \textcolor{gray!75}{837}/8.1 & \textcolor{gray!75}{1207}/10.0 \\
                             & $\boldsymbol{4}$ & 45.57/5.9 & \textcolor{gray!75}{173}/4.0 & \textcolor{gray!75}{366}/7.0 & \textcolor{gray!75}{644}/7.5 & \textcolor{gray!75}{1005}/9.0 & \textcolor{gray!75}{1451}/7.5 \\
        \hline
        \multirow{4}{*}{$\boldsymbol{4}$} & $\boldsymbol{1}$ & \textcolor{gray!75}{109}/6.3 & \textcolor{gray!75}{401}/8.2 & \textcolor{gray!75}{844}/7.9 & \textcolor{gray!75}{1489}/16.7 & \textcolor{gray!75}{2328}/35.0 & \textcolor{gray!75}{3350}/10.5 \\
                             & $\boldsymbol{2}$ & \textcolor{gray!75}{258}/12.7 & \textcolor{gray!75}{948}/9.4 & \textcolor{gray!75}{1997}/8.3 & \textcolor{gray!75}{3523}/9.5 & \textcolor{gray!75}{5504}/15.6 & \textcolor{gray!75}{7951}/9.5 \\
                             & $\boldsymbol{3}$ & \textcolor{gray!75}{313}/4.7 & \textcolor{gray!75}{1148}/8.3 & \textcolor{gray!75}{2424}/4.6 & \textcolor{gray!75}{4267}/4.7 & \textcolor{gray!75}{6665}/10.3 & \textcolor{gray!75}{9620}/10.8 \\
                             & $\boldsymbol{4}$ & \textcolor{gray!75}{378}/3.2 & \textcolor{gray!75}{1387}/6.9 & \textcolor{gray!75}{2923}/1.5 & \textcolor{gray!75}{5146}/8.2 & \textcolor{gray!75}{8021}/6.8 & \textcolor{gray!75}{11551}/9.0 \\
        \hline

    \end{tabular}

\end{threeparttable}
     \caption{Evaluation of the Maskable PPO algorithm in the parameterized level crossing example. It displays the state space size (in millions) on the left and the median convergence time (in seconds) over 20 repetitions on the right. Median was used instead of averages to reduce the impact of outliers in measurements. Entries with state space sizes colored in \textcolor{gray!75}{gray} indicate estimated measurements.}
     \label{tab:drl-results-1}
\end{table}

\vspace{-5pt}

Integrating DRL algorithms into the MDP-based approach introduces a scalable and lightweight solution for enforcing liveness in b-programs. Currently, our implementation produces a deterministic policy, unlike the standard MDP-based method that allows a non-deterministic policy as defined in \autoref{def:Q-compatible policy}. Our experiment highlights that while this approach has limitations, it still holds significant value as it enables the execution of live specifications in scenarios where exhaustive state space exploration is impractical. 


\section{multiple vs. single liveness requirements}
\label{sec:examples}

This section evaluates the execution mechanisms by considering a use case that involves multiple liveness requirements. We first present a Sokoban game version that includes a single liveness requirement. Then, we generalize this version to accommodate multiple requirements. We use the generalized version to assess and illustrate how the GBA and MDP-based approaches effectively handle single and multiple requirement scenarios, demonstrating their efficiency.

Sokoban is a classic puzzle where a player controls a warehouse worker who pushes boxes to target locations while constrained by walls. Since boxes can only be pushed, not pulled, many moves are irreversible, and mistakes can lead to a state from which there is no solution. We refer to such scenarios as ``live locks.'' Sokoban has been extensively studied in the fields of model checking and synthesis, as demonstrated in various works~\cite{kant_ltsmin_2015, vinkhuijzen_symbolic_2020}. This paper focuses on demonstrating how liveness requirements can be used to prevent entering live locks. Specifically, we require that ``a box should be in a target location infinitely often.'' This ensures that we never move to a state with a box we cannot bring to a target location. Note that this condition is necessary and sufficient for a prefix of a live run, so we are not over-specifying. The Sokoban b-program and all tools used for evaluation in this section are available in the supplementary material.

We examined two possible liveness objectives for the model. The first is that ``all boxes are simultaneously placed at target locations infinitely often.'' This can be expressed as a single liveness requirement, modeled in the b-thread in \autoref{lst:sokoban-box-liveness-single}.
 
\begin{lstlisting}[
label={lst:sokoban-box-liveness-single},
caption={A b-thread specifying that ``all boxes are simultaneously placed at target locations infinitely often.''},
% float=htbp
]
def all_boxes_in_target(boxes, targets):
  while True:
    all_tar = all([box in targets for box in boxes] 
    yield sync(waitFor=All(),mustFinish=not all_tar)
\end{lstlisting}

We also examine the relaxed liveness requirement: ``Each box must be placed in a target location infinitely often" (not necessarily simultaneously). We replaced the b-thread in \autoref{lst:sokoban-box-liveness-single} to implement this requirement with the b-thread in \autoref{lst:sokoban-box-liveness-multiple}. This b-thread is added for each box in the program.

\begin{lstlisting}[
label={lst:sokoban-box-liveness-multiple},
caption={A b-thread specifying that ``Each box is placed at a target location infinitely often.''},
% float=htbp
]
@b_thread
def box(k, boxes, targets):
  while True:
    in_tar = boxes[k] in targets
    yield sync(waitFor=All(), mustFinish=not in_tar)
\end{lstlisting}

We conducted experiments on 20 randomly generated Sokoban boards of varying sizes containing 1-3 boxes, using open-source software based on the algorithm presented in~\cite{taylor_procedural_2011}. For the GBA-based approach, we used the degeneralization algorithm~\cite{babiak_compositional_2013} and solver implemented in the Spot library~\cite{duret-lutz_spot_2016}. For the MDP-based approach, we used the value-iteration algorithm implemented in the MDPtoolbox library~\cite{chades_mdptoolbox_2014} to solve the Bellman equation. The b-program was translated into a GBA and MDP in a pre-processing step using a depth-first search algorithm. Experiments were conducted using an Intel Xeon E5-2620 CPU with 64GB RAM.

The runtime of the evaluated method, depicted in \autoref{fig:g}, is presented as a function of each board's program state space size. In the case of the single liveness scenario, the graph in \autoref{fig:g1} shows that the automata-based method has a slight advantage over the MDP-based approach while keeping a consistent ratio, suggesting that the discrepancy is probably due to implementation. Both methods effectively handle relatively large graphs. Notably, the non-monotonicity in the graph highlights that the program's state space size is not the sole determining factor for complexity. This aligns with expectations, as the number of boxes also contributes to the overall complexity.
\vspace{-5pt}
\begin{figure}[ht]
     \begin{subfigure}{0.55\textwidth}
    \begin{tikzpicture}[scale=0.7, transform shape]
    \pgfplotsset{%
    width=1.25\textwidth,
    height=0.45\textwidth
    }
    \pic {graph_all};
    \end{tikzpicture}
    \caption{Single}
    \label{fig:g1}
     \end{subfigure}
     \qquad\qquad
     \begin{subfigure}{0.55\textwidth}
    \begin{tikzpicture}[scale=0.7, transform shape]
    \pgfplotsset{%
    width=1.25\textwidth,
    height=0.45\textwidth
    }
    \pic {graph_per_bt};
    \end{tikzpicture}
    \caption{Multiple}
    \label{fig:g2}
     \end{subfigure}
    \caption{The runtime of the GBA and MDP-based approaches for the single (\autoref{fig:g1}) and multiple (\autoref{fig:g2}) liveness requirements scenarios. Measurements are presented as a function of the program's state-space size.}
    \label{fig:g}
\end{figure}
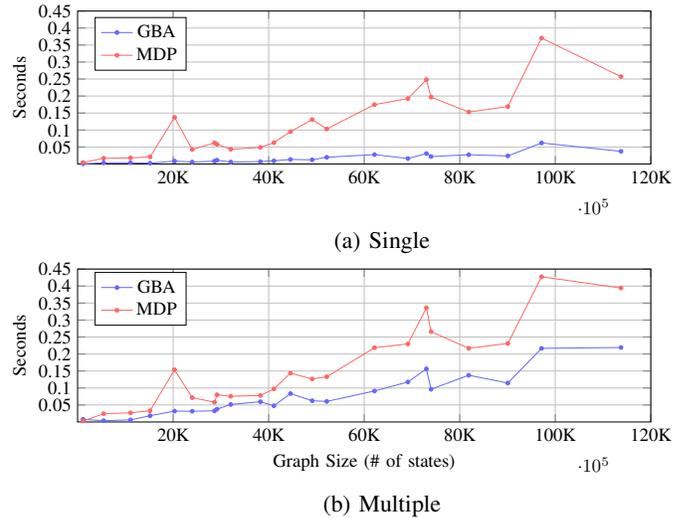
As for the multiple liveness requirements scenario, we applied a heuristic to the MDP-based approach to address the reward-based mechanism, where the calculated reward is the sum of rewards for each requirement. Although this technique has not been formally proven for correctness, it showed improved results in practice, as results presented in \autoref{fig:g2} reveal that transitioning from a single to multiple liveness requirements did not significantly impact the MDP-based approach's performance. This improvement can be attributed to the heuristic's ability to eliminate the necessity of expanding the state space. This step would have been required if an additional b-thread had been introduced, as suggested in~\cite{baier_principles_2008}, without affecting performance. Further, the results presented in \autoref{fig:g2} demonstrate a degradation in the GBA-based approach's performance when transitioning from a single to multiple liveness requirements. This increased runtime can be attributed to the degeneralization process used when multiple acceptance sets, each representing a different requirement, are combined.

\subsection{Resilience to errors}

In the theoretical part of this study, we considered an exact solution to the MDP, represented by $Q^*$. Nevertheless, even when the MDP is known, algorithms such as value iteration only provide an approximation for $Q^*$. Clearly, this is fairly accurate compared to scenarios where the MDP is unavailable, as discussed in \autoref{sec:drl}; however, minor perturbations can still lead to errors, especially when applying the above-mentioned reward sum heuristic. To ensure the practicality of this approach, we assessed its robustness by evaluating the performance of the $Q^*$-compatible event selection mechanism under increasing levels of Gaussian noise added to the estimated $Q$ values. The metric used was the live-run rate, representing the frequency of runs without a live lock occurrence. A boxplot aggregating the live-run rate for the 20 Sokoban examples is displayed in \autoref{fig:bp1}. We observe that the reward sum heuristic demonstrates resilience to errors and loses very little precision compared to the single liveness requirement scenario.

\vspace{-5pt}
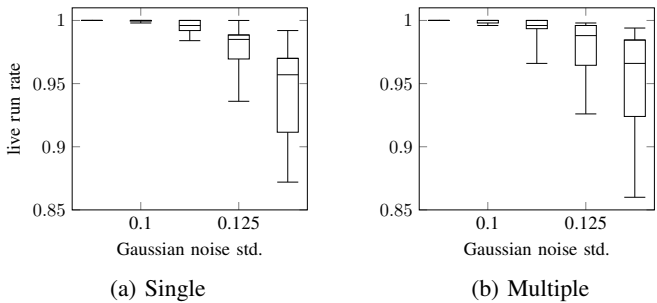
\begin{figure}[ht]
     \begin{subfigure}{0.23\textwidth}
    \begin{tikzpicture}[scale=0.7, transform shape]
    \pgfplotsset{%
    width=1.45\textwidth,
    height=1.3\textwidth
    }
    \pic {box_plot_all};
    \end{tikzpicture}
    \caption{Single}
    \label{fig:bp1}
     \end{subfigure}
     \qquad
     \begin{subfigure}{0.23\textwidth}
    \begin{tikzpicture}[scale=0.7, transform shape]
    \pgfplotsset{%
    width=1.45\textwidth,
    height=1.3\textwidth
    }
    \pic {box_plot_per_bt};
    \end{tikzpicture}
    \caption{Multiple}
    \label{fig:bp2}
     \end{subfigure}
     \caption{Performance of the $Q^*$-compatible event selection mechanism for the single (\autoref{fig:bp1}) and multiple (\autoref{fig:bp2}) liveness requirements scenarios under increasing levels of Gaussian noise added to the estimated $Q$ values. The metric used was the live-run rate, representing the frequency of runs without a live lock occurrence out of 1000 sampled runs.}
        \label{fig:bp}
\end{figure}
\vspace{-5pt}

\section{related work}
\label{sec:related-work}
As mentioned in \autoref{sec:motivation}, Harel et al.~\cite{harel_model-checking_2011} presented a methodology and a supporting model-checking tool for verifying b-programs. 
In this methodology, b-threads can mark states as hot, and the model checker verified liveness properties using a nested DFS algorithm, searching for cycles in the state graph that contain only hot states. That means that the verified liveness property is that the b-program is always eventually not in a hot state. Bar-Sinai~\cite{bar-sinai_extending_2020} later distinguished between cycles in which a single b-thread is indefinitely hot and cycles where the b-program, as a whole, is indefinitely hot.

Our definition of liveness is inspired by both~\cite{harel_model-checking_2011} and~\cite{bar-sinai_extending_2020}. The work in this paper extends the above verification approach to synthesizing liveness-enforcing event selection mechanisms for BP. While similar research has been carried out for other languages (e.g.~\cite{dippolito_synthesizing_2013, uzam_improved_2006}), these languages are not focused on a direct alignment with requirements. This extension elevates BP and brings it closer to its goal of aligning with requirements. To the best of our knowledge, the work presented here is the first that has been made on liveness enforcing execution mechanism of b-programs. Further, it is the first to propose an alternative to the automata-based approach, based on MDP formulation, which is especially useful when it is computationally intractable to guarantee liveness.

 

General smart execution methods for BP were also introduced in~\cite{harel_smart_2002, eitan_adaptive_2011, weinstock_online_2015, elyasaf_using_2019}. These methods are not aimed at liveness requirements specifically but provide some foundations on which this paper relies. 
In~\cite{eitan_adaptive_2011}, Eitan and Harel extended the semantics of BP with reinforcements, allowing applications that specify broad goals, in addition to what should be done or not done at every step. Reinforcements are captured in~\cite{eitan_adaptive_2011} by b-threads, each contributing a narrow assessment of the current situation relative to a longer-term goal. 
Learning and adapting allow programmers to focus on specifying needs while leaving the details and optimizations to the program, thus simplifying development. The MDP-based approach suggested in this paper continues, in some ways, the work of~\cite{eitan_adaptive_2011}. Our approach extends the capability to allow the integration of deep reinforcement learning and ensures liveness in BP.
In another work, Weinstock~\cite{weinstock_online_2015} extended the BP execution mechanism with an online heuristic search in program state space that allows programmers to develop programs while relying on a smart event selection mechanism to resolve non-determinism in a way that maximizes a defined heuristic function.

Significant efforts have been made toward translating logical specifications into reinforcement learning. In this line of research, reinforcement learning is being used to automatically construct optimal policies concerning linear temporal logic (LTL) specifications~\cite{alur_framework_2022, camacho_ltl_2019}. We note that reinforcement learning policies are executable but hard to explain and maintain, while LTL possesses the opposite characteristics. Hence, practicing this approach will require going back and forth in each update, resulting in possibly prolonged development cycles. Along these lines, the same can be argued about the direct formulation of the MDP or GBA. 


\section{conclusion}
\label{sec:conclusion}



We have demonstrated how liveness requirements can be directly expressed in the BP executable modeling paradigm by adding the ``must-finish'' idiom. 
The outcome of this study is comprehensive execution mechanisms for b-programs, capable of generating runs that meet all requirements, including liveness. Our method can also be used during development for the early identification of liveness conflicts. The BP approach is particularly suitable for such incremental refinement, as it supports refinement by adding new b-threads without changing existing ones~\cite{harel_programming_2010}. While the primary focus is on BP, the concepts presented here can be extended to other executable-specification languages. Specifically, the ``must-finish'' idiom and the mapping of a specification into a GBA or MDP, which targets the compliance of liveness requirements, are general and can be adapted to other languages.

\section{Data Availability}
The complete code for the evaluations and the tools we have implemented are publicly available at \url{https://anonymous.4open.science/r/bp-liveness-4C6F}.



%


\bibliographystyle{IEEEtran}
\bibliography{references,more}

\clearpage
\appendix[Proof of \autoref{theorem:ab}]
\label{app:a}

We now formulate and prove a series of claims and propositions towards proving \autoref{theorem:ab}.
To simplify the proofs, we assume that all initial states of all the b-threads are not labeled as \emph{must-finish}: 
\begin{assumption}
\label{assumption:a}
$L(init) = 0$.
\end{assumption}
As demonstrated with our examples in \autoref{sec:motivation} and \autoref{sec:examples}, this assumption is reasonable in practice. It does not restrict generality since adding an extra initial state that is not labeled as \emph{must-finish} is always possible. 

The first claim says that the accumulated rewards over a finite prefix of a run are either $0$ or $-1$:

\begin{claim}
\label{claim:a}
For every infinite b-program run $l=s^0 \xrightarrow{e^0} s^1 \xrightarrow{e^1} \cdots $ and time $t\geq 0$:
$$\sum_{k=0}^{t}R( s^{k-1} ,e^{k-1}, s^{k}) = 
\begin{cases}
  0  & \text{if } L(s^t)=0;\\
 -1  & \text{otherwise. }\\
\end{cases}
$$ 
\end{claim}
\begin{proof}
By induction on $t$. The base case is given by \autoref{assumption:a}. 
Assuming that the claim is true for $t-1$, If $L(s^{t-1})=L(s^t)$ then $R(s^{t-1},e^{t-1},s^t)=0$ and the claim follows. If $L(s^{t-1})=0$  and $L(s^t)=1$  then $R(s^{t-1},e^{t-1},s^t)=-1$  and we get that 
$$\sum_{k=0}^{t}R( s^{k-1} ,e^{k-1}, s^{k}) = \sum_{k=0}^{t-1}R( s^{k-1} ,e^{k-1}, s^{k})  - 1 = -1.$$
If $L(s_i^{t-1})=1$  and $L(s_i^t)=0$  then $R(s^{t-1},e^{t-1},s^t)=1$  and we get that $$\sum_{k=0}^{t}R( s^{k-1} ,e^{k-1}, s^{k}) = \sum_{k=0}^{t-1}R( s^{k-1} ,e^{k-1}, s^{k})  + 1 = 0.$$
Hence, all cases are consistent with the claimed equation.
\end{proof}

We next show that the sequence of rewards is a repetition of the form:
$0,\dots,0,-1,0,\dots,0,1$ where $0,\dots, 0$ means, possibly empty, sequence of $0$s. There may be an infinite tail of zeroes at the end, or the alternation can go forever.

\begin{claim}
\label{claim:b}
For every infinite b-program run $l=s^0 \xrightarrow{e^0} s^1 \xrightarrow{e^1} \cdots $, let $(t_k)_{k=0}^{n}$ be the sequence of times where $R(s^{t_k},e^{t_k},s^{t_{k}+1})\neq 0$. The length of the sequence can be finite, infinite, or empty, i.e., $n \in \mathbb{N} \cup \{\infty, -1\}$. Then for every $0 {\leq} k {\leq} n\colon R(s^{t_k},e^{t_k},s^{t_{k}+1}) =  (-1)^{k+1}$.
\end{claim}
\begin{proof}
Based on \autoref{def:bpmdp} and \autoref{assumption:a}, it is clear from that $R(s^{t_0},e^{t_0},s^{t_{0}+1})=-1$. Assume towards contradiction that there is $k{\geq} 0$ such that
$R(s^{t_k},e^{t_k},s^{t_{k}+1})  = R(s^{t_{k+1}},e^{t_{k+1}},s^{t_{k+1}+1})$.
Since $R(s^{t},e^{t},s^{t+1})=0$  for every $t_k < t < t_{k+1}$, by the same definition, $L(s^{t_{k}+1}) = L(s^{t_{k+1}})$. This contradicts the definition of $R$ where it is apparent that  $ L(s^{t_{k}+1}) = L(s^{t_{k+1}})$ 
implies
$R(s^{t_k},e^{t_k},s^{t_{k}+1})  \neq R(s^{t_{k+1}},e^{t_{k+1}},s^{t_{k+1}+1})$.
\end{proof}
 
Using the above observation regarding the alternation of the sequence, we obtain a lower bound for the residual discounted accumulated reward of live runs based on the current state's label:
 
\begin{claim}
\label{claim:c}
For every infinite live b-program run $l=s^0 \xrightarrow{e^0} s^1 \xrightarrow{e^1} \cdots $, time $t\geq 0$, and $\gamma < 1$:
$$\sum_{k=t}^{\infty}\gamma^{k} R(s^{k},e^{k},s^{k+1}) > 
\begin{cases}
  -1  & \text{if }L(s^{t})=0;\\
 0  & \text{if }L(s^{t})=1.\\
\end{cases}
$$ 
\end{claim}
\begin{proof}
Similarly to the sequence used in \autoref{claim:b}, let $(q_k)_{k=0}^{n_q}$ be the sequence of times after $t$  where $R(s^{q_k},e^{q_k},s^{q_k+1})=1$, and $(r_k)_{k=0}^{n_r}$ be the sequence of times after $t$ where $R(s^{r_k},e^{r_k},s^{r_k+1})=-1$. Note that since run $l$ is live, we have that $n_q \geq n_r$; otherwise, the run ends with infinitely many b-program's must-finish states. 
If both sequences are empty, it is clear from \autoref{def:bpmdp} and \autoref{assumption:a} that $L(s_i^t)=0$. In this case, all rewards are zero, and the claim holds trivially. Furthermore, if $L(s^t)=1$ and $n_r < 0$, then $(q_k)_{k=0}^{n_q}$ is not empty, i.e., $n_q \geq 0$ or else the run ends with infinitely many must-finish states. In this case, we get that
$$\sum_{k=t}^{\infty}\gamma^{k} R(s^{k},e^{k},s^{k+1}) =\sum_{k=0}^{n_q}\gamma^{q_k} > 0.$$
If the sequences are not empty and $L(s_i^t)=1$, from \autoref{claim:b} we get that $q_k < r_k$ for each $k \leq n_r$ and then
$$\sum_{k=t}^{\infty}\gamma^{k} R(s^{k},e^{k},s^{k+1}) \ge\sum_{k=0}^{n_r}(\gamma^{q_k} - \gamma^{r_{k}})> 0.$$
If the sequences are not empty and $L(s_i^t)=0$, from \autoref{claim:b} we get that $r_{k} < q_{k} < r_{k+1}$ for each $k<n_r-1$ and $$\sum_{k=t}^{\infty}\gamma^{k} R(s^{k},e^{k},s^{k+1}) \ge- \gamma^{r_{0}}+\sum_{k=0}^{n_r-1}(\gamma^{q_k} - \gamma^{r_{k+1}})  > -1.$$
Thus, the claim holds for all cases.
\end{proof}

In the opposite direction to the previous claim, we also show that if the optimal policy can achieve a positive residual discounted accumulated reward, it is possible to get to a state that is not labeled as \emph{must-finish} (we will later use that to construct a live run).

\begin{claim}
\label{claim:d}
If $Q^*(s^{t},e^{t}) > 0$ then there is a path $s^t \xrightarrow{e^t} s^{t+1} \xrightarrow{e^{t+1}} \cdots \xrightarrow{e^{t+m_{t-1}}} s^{t+m_t}$  such that $L(s^{t+m_t})=0$.
\end{claim}
\begin{proof}
Using the optimal policy $\pi^*$, we construct a path by defining $e^{t'} = \pi^*(s^{t'})$ for every $t' > t$, and choosing $s^{t'+1}$ to be the only state such that $P(s^{t'},e^{t'},s^{t'+1})=1$. There is only one such state since the b-program transitions (as defined in \autoref{def:run}) are deterministic. Assume, towards contradiction, that $L(s^{t'}) = 1$ for every $t' \geq t$. Then  
$Q^*(s^{t},e^{t})=\sum_{t'=t}^{\infty}\gamma^{t'} R(s^{t'},e^{t'},s^{t'+1}) =0$, 
which contradicts the assumption. This gives us that the path that we have constructed is as required.
\end{proof}

We are now ready to state and prove the two propositions that establish the correctness of our approach, starting with showing that an execution mechanism that generates all $Q^*$ compatible runs is complete in the sense that it generates all possible live runs:

\begin{proposition}
\label{prop:a}
A live b-program run is $Q^*$-compatible.
\end{proposition}
\begin{proof}
Let $l = s^0 \xrightarrow{e^0} s^1 \xrightarrow{e^1} \cdots $ be a live run. To prove that $l$ is $Q^*$-compatible we now show that the term in \autoref{def:Q-compatible run} holds for every time $t$. 

If $L(s^{t}) = 0$, from \autoref{claim:a} we get that 
$\sum_{k=0}^{t}R( s^{k} ,e^{k}, s^{k+1}) = 0$
and, as shown in \autoref{claim:c}
$\sum_{k=t}^{\infty}\gamma^{k} R(s^{k},e^{k},s^{k+1})  > -1$.

If $L(s^{t}) = 1$, from \autoref{claim:a} we get that
$\sum_{k=0}^{t}R( s^{k} ,e^{k}, s^{k+1}) = -1$
and, as shown in \autoref{claim:c}
$\sum_{k=t}^{\infty}\gamma^{k} R(s^{k},e^{k},s^{k+1})  > 0$.

In both cases, when adding the terms together, we get that for every time $t$
$\sum_{k=0}^{t}R( s^{k} ,e^{k}, s^{k+1}) + \sum_{k=t}^{\infty}\gamma^k R(s^{k},e^{k},s^{k+1}) > -1$.

By the definition of $Q^*$, the optimal-value function,\\
$\sum_{k=0}^{t}R( s^{k} ,e^{k}, s^{k+1}) + Q^*(s^t,e^t) > -1$.

We get that run $l$ is $Q^*$-compatible by \autoref{def:Q-compatible run}.
\end{proof}

Second, we show that an execution that generates $Q^*$ compatible runs according to the distribution defined in \autoref{def:Q-compatible policy} is sound in the sense that it generates live runs with probability one.

\begin{proposition}
\label{prop:b}
A $Q^*$-compatible b-program run is almost surely live.
\end{proposition}
\begin{proof}
Let $\pi$ be the policy defined in \autoref{def:Q-compatible policy} using the optimal value function $Q^*$. We will show that $\pi$ generates a live run with probability one. Since, by definition, $\pi$ always draws a $Q^*$ compatible runs, we have 
$\sum_{k=0}^{t-1}R( s^{k} ,e^{k}, s^{k+1} ) + Q^*(s^{t},e^{t}) > -1$
for all $t\geq 0$. To generate a non-live run, $\pi$ needs from some point of time, $t_0$, to always visit states that are labeled as \emph{must-finish}, i.e., $L(s^{t}) = 1$ for all $t > t_0$. By~\autoref{claim:a}, for all $t > t_0$, 
$\sum_{k=0}^{t-1}R( s^{k} ,e^{k}, s^{k+1}) = -1$
and we get that 
$Q^*(s^{t},e^{t}) > 0$.

Therefore, by~\autoref{claim:d} for every $t > t_0$ there is a path
$s^t \xrightarrow{\hat{e}^t} \hat{s}^{t+1} \xrightarrow{\hat{e}^{t+1}} \cdots \xrightarrow{\hat{e}^{t+m_{t-1}}} \hat{s}^{t+m_t}$  such that $L(\hat{s}^{t+m_t})=0$. Assume that this is the first such index, i.e., $L(\hat{s}^{t+m_t-1})=1$ and we get that 
$\sum_{k=t}^{t'-1}R( \hat{s}^{k} ,\hat{e}^{k},\hat{s}^{k+1} )=0$
for every $t \leq t' < t+m_t-1$.  This means that all the states along this path satisfy 
$$\sum_{k=0}^{t-1}R( s^{k} ,e^{k}, s^{k+1} )  + \sum_{k=t}^{t'-1}R( \hat{s}^{k} ,\hat{e}^{k},\hat{s}^{k+1} ) + Q^*(\hat{s}^{t'},\hat{e}^{t'}) > -1$$ 
and that $\pi$ could have chosen this path. The probability that it will not choose any of these paths is zero.
\end{proof}

Finally, we get that \autoref{theorem:ab} holds, a live b-program run is $Q^*$-compatible, and a $Q^*$-compatible b-program run is almost surely live:
\begin{proof}[Proof of \autoref{theorem:ab}]
Follows by \autoref{prop:a} and \autoref{prop:b} above.
\end{proof}

\end{document}